\documentclass[twoside,bezier,12pt, reqno]{amsart}
\usepackage{amssymb,latexsym,epsfig,mathrsfs,graphpap,graphics,amsmath,amssymb}
\usepackage{amstext}
\usepackage{amsthm}
\usepackage[utf8]{inputenc}
\usepackage[english]{babel}
\usepackage{helvet}
\usepackage{courier}

\newtheorem{theorem}{Theorem}[section]
\newtheorem{lemma}{Lemma}[section]

\newtheorem{corollary}{Corollary}[section]

\theoremstyle{Definition}
\newtheorem{definition}{Definition}[section]

\theoremstyle{remark}
\newtheorem{remark}[theorem]{Remark}
\numberwithin{equation}{section}

\begin{document}

\begin{flushleft}
 {\bf\Large {Donoho-Stark's and Hardy's Uncertainty Principles for the  Short-time Quaternion Offset Linear Canonical Transform}}

\parindent=0mm \vspace{.2in}

{\bf{M. Younus Bhat$^{1},$ and Aamir H. Dar$^{2}$ }}
\end{flushleft}

{{\it $^{1}$ Department of  Mathematical Sciences,  Islamic University of Science and Technology Awantipora, Pulwama, Jammu and Kashmir 192122, India.E-mail: $\text{g gyounusg@gmail.com}$}}

{{\it $^{2}$ Department of  Mathematical Sciences,  Islamic University of Science and Technology Awantipora, Pulwama, Jammu and Kashmir 192122, India.E-mail: $\text{ahdkul740@gmail.com}$}}

\begin{quotation}
\noindent
{\footnotesize {\sc Abstract.} The quaternion offset linear canonical transform (QOLCT) which is time-shifted and frequency-modulated version of the quaternion linear canonical
transform (QLCT) provides a
more general framework of most existing  signal
processing tools. For the generalized QOLCT, the classical Heisenberg's and Lieb's uncertainty principles have been studied recently.  In this paper, we first define  the short-time quaternion offset linear canonical transform (ST-QOLCT)  and drive its
relationship with the quaternion Fourier transform
(QFT). The crux of the paper lies in the generalization of  several well known  uncertainty principles for the ST-QOLCT, including Donoho-Stark’s uncertainty principle, Hardy’s uncertainty principle, Beurling’s uncertainty principle,
 and
Logarithmic uncertainty principle.
 \\


}
\end{quotation}
\section{ \bf Introduction}
\noindent
The linear canonical transform (LCT) with four parameters $(a, b, c, d)$ has been generalized to a six parameter transform  $(a,b,c,d,u_0,w_0)$ known as offset linear canonical transform (OLCT). Due to the time shifting $u_0$ and frequency modulation parameters, the OLCT has gained more flexibility over classical LCT. Hence has found wide applications in image and signal processing.  The quaternion offset linear canonical transform (QOLCT) which is time-shifted and frequency-modulated version of the quaternion linear canonical transform (QLCT) provides a more general framework of most existing  signal processing tools. For more details we refer to \cite{{16}, {14}, {30}} and references therein.

Because of its wide applications in signal analysis, image processing and optics the quaternion offset linear canonical transform(QOLCT) has attained much universality in recent years. However, the QOLCT is inadequate for localizing the QOLCT-frequency of non-transient signals, as such, it is indispensable to introduce an eccentric localized transform coined as the short-time quaternion offset linear canonical transform(ST-QOLCT), which can effectively reveal the local QOLCT-frequency content of such signals. The ST-QOLCT enjoys high resolution, provides local Information and eliminates cross terms. The chrip signals can be better analysed through ST-QOLCT. We refer to \cite{{OWN1}, {8},  {gen}} for more details.

Let us now move to the side of uncertainty inequality. Uncertainty principle was introduced by German physicists Heisenberg \cite{hsb} in 1927 which is known as the  heart of any signal processing tool. With the passage of time researchers further extended the uncertainty principle to different types of new uncertainty principles associated with the Fourier Transform, for instance Heisenberg’s uncertainty principle, Logarithmic uncertainty principle,
Hardy’s uncertainty principle and Beurling’s uncertainty principle. Later these uncertainty principles were extended to Quaternion domain. In \cite{14,30} authors proposed uncertainty principles associated with the
OLCT and in \cite{8,19,WOLCT} authors establish uncertainty principles for the WLCT and WOLCT. Recently, the uncertainty principles associated with the
QOLCT were proposed in \cite{gen,feb23}. Also Gao and Li\cite{WLCT} recently developed uncertainty principles for two sided windowed linear canonical transform. Later M.Y Bhat and A.H Dar\cite{OWN1} establish uncertainty principles for 2D Gabor quaternion offset linear canonical transform.  Where as Lieb's uncertainty principle has been established in \cite{feb23}.  However,  Donoho-Stark's uncertainty principle,  Hardy's uncertainty principle  and Beurling's uncertainty principle have not been established for ST-QOLCT. Taking this opportunity,  we shall study  these uncertainty principles for the ST-QOLCT domain. 

The rest of paper is organised as follows. In Section 2, we provide some preliminaries needed for subsequent sections. In Section 3, we establish a relationship of ST-QOLCT with  QOLCT and QFT. In Section 4, we develop some novel uncertainty principles like Donoho-Stark's, Hardy's  and Beurling's. Finally we establish Logarithmic uncertainty principle using Pitt's Inequality.

\section{\bf Preliminaries}
\label{sec 2}
In this section,  we collect some basic facts on the quaternion algebra and the QFT, which will be needed throughout the paper.

\subsection{\bf Quaternion algebra} \  \\

In 1834 W. R. Hamilton introduced quaternion algebra by extension of the complex number to  an associative non-commutative 4D algebra. Denoted by $\mathbb{H}$ in his honor where every element of $\mathbb{H}$ has a Cartesian form given by
\begin{equation}\label{eqn 2.1}
\mathbb{H}=\left\lbrace q|q:=[q]_0+{i}[q]_1+{j}[q]_2+{k}[q]_3, [q]_i\in\mathbb{R}, i=0,1,2,3\right\rbrace
\end{equation}
where ${i}, {j}, {k}$ are imaginary units obeying Hamilton's multiplication rules:
\begin{equation}\label{eqn 2.2}
		{i}^2={j}^2={k}^2=-1,\\
\end{equation}

\begin{equation}\label{eqn 2.3}
{i}{j}=-{j}{i}={k},\,{j}{k}=-{k}{j}={i},\,{k}{i}
=-{i}{k}={j}.
\end{equation}
Let $[q]_{0}$ and $q={i}[q]_1+{j}[q]_2+{k}[q]_3$ denote the real scalar part and the vector part of quaternion number $q=[q]_0+\mathbf{i}[q]_1+{j}[q]_2+{k}[q]_3$, respectively. Then, from \cite{53}, the
real scalar part has a cyclic multiplication symmetry
\begin{equation}\label{eqn 2.4}
		[pql]_{0}=[qlp]_{0}=[lpq]_{0}, \qquad \forall q,p,l\in \mathbb{H},
\end{equation}
the conjugate of a quaternion $q$ is defined by $\overline{q}=[q]_0-{i}[q]_1-{j}[q]_2-{k}[q]_3$, and the norm of $q\in \mathbb{H}$ defined as
\begin{equation}\label{eqn 2.5}
\left| q\right|=\sqrt{q\bar{q}}=\sqrt{[q]_0^2+[q]_1^2+[q]_2^2+[q]_3^2}{ .}
\end{equation}
It is easy to verify that
\begin{equation}\label{eqn 2.6}
		\overline{pq}=\overline{q}  \overline{p}, \qquad |qp|=|q||p|, \quad \forall q,p\in \mathbb{H}.
\end{equation}
In this paper, we will study the  quaternion-valued signal $f:{\mathbb{R}}^2\to \mathbb{H} $, $f$ which can be expressed as $f=f_0+i f_1+jf_2+kf_3,$ with $f_m~:  {\mathbb R}^2 \to \ {\mathbb R}\ for\ m=0,1,2,3.$
 The quaternion inner product for quaternion valued signals $f,g\ :{\mathbb R}^2\ \to {\mathbb H}$, as follows:
 \begin{equation}\label{eqn 2.7}
\langle f,g\rangle = \int_{{\mathbb R}^2} {f\left(\mathbf x\right)\overline{g\left(\mathbf x\right)}}d\mathbf x
\end{equation}

where $\mathbf x=(x_1,x_2),$   $f(\mathbf x)=f(x_1,x_2),$ $\mathbf x=dx_1 dx_2,$ and so on.\\
Hence, the natural norm is given  by
\begin{equation}\label{eqn 2.8}
{\left|f\right|}_{2}=\sqrt{<f,f>} ={(\int_{{\mathbb R}^2}{{\left|f(\mathbf x)\right|}^2}d\mathbf x)}^{\frac{1}{2}},
\end{equation}
and the quaternion module $L^2({\mathbb R}^2,\ {\mathbb H})$, is given by
\begin{equation}\label{eqn 2.9}
L^2({\mathbb R}^2,\ {\mathbb H}) = \{f : {\mathbb {\mathbb R}^2\ \to {\mathbb H},\ {\left|f\right|}_{2}< \infty }\}.
\end{equation}

\begin{lemma}\label{lem 2.1}
If $f,g\in L^2(\mathbb{R}^2,\mathbb{H})$, then the Cauchy-Schwarz inequality holds[25]
\begin{equation}\label{2.10}
	\left| \langle f,g\rangle_{L^2(\mathbb{R}^2,\mathbb{H})}\right|^2\leq \|f\|_{L^2(\mathbb{R}^2,\mathbb{H})}^2\|g\|_{L^2(\mathbb{R}^2,\mathbb{H})} ^2.
\end{equation}
If and only if $f=i g$ for some quaternionic parameter $i \in\mathbb{H}$, the equality holds. \end{lemma}

\parindent=0mm \vspace{.4in}

\subsection{The two-sided QFT } \  \\

The QFT belongs to the family of Clifford Fourier transformations.

 It is a generalization of the classical Fourier transform (CFT)  \cite{EL93}.Some useful properties, and theorems of this  transform are
generalizations of the corresponding properties and theorems of the classical Fourier transform with some
modifications.
There are three different types of QFT, the left-sided QFT , the right-sided QFT, and two-sided QFT \cite{PDC01}. In this paper our focus shall be on two-sided QFT. So from  here on  by QFT we mean two-sided quaternion Fourier transform.
Let us  begin with definition of the two-sided QFT and provide some properties used in the sequel.
\begin{definition}\label{def 2.1}
(Two-sided QFT.)\\
For  $f\in L^1\left({\mathbb R^2},{\mathbb H}\right)$ the two-sided QFT with respect to unit quaternions $i ; j $ is given by\\
\begin{equation}\label{QFT}
{\mathcal F}^{i, j }[f](\mathbf w) =\int_{{\mathbb R}^2}{e^{-i {{ w}}_1x_1}}\ f(\mathbf x)\ e^{-j {{ w}}_2x_2}dt,~~ where ~\mathbf{x, w}\in {\mathbb R}^2.
\end{equation}
\end{definition}

We define the modulus of $\mathcal F[f]^{i, j } $ \ as follows :
\begin{equation}\label{modulus}
{\left|{\mathcal F}^{i, j }[f]\right|}\ := \sqrt{\sum^{m=3}_{m=0}{{\left|{\mathcal F}^{i, j }\left[f_m\right]\right|}^2}}.\end{equation}

Furthermore, we define a new $L^2$-norm of ${\mathcal F[f]}$ as follows :\\
\begin{equation}\label{norm}
{\left\|{\mathcal F}^{i, j }[f]\right\|}_{2}:=\sqrt{\int_{\mathbb {\mathbb R}^2}{{\left|{\mathcal F}^{i, j }\left[f\right](y)\right|}^2dy}}.\end{equation}

\begin{lemma}\label{dilation}{\bf(Dilation property)} \ \\
Let $k_1, k_2$ be a positive scalar constants,  we have
\begin{equation}
{\mathcal F}^{i, j }\left[f(x_1,x_2)\right]\left(\frac{w_1}{k_1},\frac{w_2}{k_1}\right)=k_1k_2{\mathcal F}^{i, j }\left[f(k_1x_1,k_2x_2)\right]\left(w_1,w_2\right).
\end{equation}
we can also write it as
\begin{equation}
{\mathcal F}^{i, j }\left[f(\mathbf x)\right]\left(\frac{\mathbf w}{\mathbf k}\right)=\mathbf k{\mathcal F}^{i, j }\left[f(\mathbf {kx})\right]\left(\mathbf w\right).
\end{equation}
\end{lemma}

\begin{lemma}\label{Plancherel}{\bf(QFT Plancherel)}\\
Let $f \in L^2({\mathbb R}^2,{\mathbb H})$, then
\begin{equation}
\int_{{\mathbb R}^2}{{\left|{\mathcal F}^{i,j }\left[f\right]\left(\mathbf w\right)\right|}^2}d\mathbf w=4{\pi }^2\int_{{\mathbb R}^2}{{\left|f(\mathbf x)\right|}^2}d\mathbf x.
\end{equation}
\end{lemma}

\begin{lemma}\label{inverse}{\bf (Inverse QFT)}

If  $f\in L^1\left({\mathbb{R}}^2, {\mathbb{H}}\right), and \ {\mathcal F}^{i,j}[f]\in L^1\left({\mathbb{R}}^2,{\mathbb{H}}\right)$, then the two-sided QFT is an invertible transform and its inverse is given
by\\
\begin{equation}
f(\mathbf x)= \frac{1}{{(2\pi )}^2} \int_{{\mathbb{R}}^2}{e^{i w_1 x_1}} {\mathcal F}^{i,j} [f(\mathbf x)](\mathbf w) e^{j w_2 x_2}d\mathbf w.
\end{equation}
\end{lemma}
\subsection{\bf Quaternion offset linear canonical transform (QOLCT) } \ \\
The quaternion linear canonical transform(QLCT) is a generalization of the linear canonical transform(LCT) firstly defined by Kou, et al. \cite{50,16} . Later in \cite{Hitzer2}  Hitzer.E generalized the definitions of  Kou, etl to introduce two-sided QLCT. In this paper, we mainly focus on the two-sided QLCT.
\begin{definition}{\bf(Quaternion Linear Canonical Transform.)}\ \\
 Let
	$A_s=\begin{bmatrix}
	a_s&b_s\\
	c_s&d_s
	\end{bmatrix}\in \mathbb{R}^{2\times 2}$ be a matrix parameter satisfying ${\rm det}(A_s)=1$, for $s=1,2$. The two-sided QLCT of signal $f\in L^2\left( \mathbb{R}^2,\mathbb{H}\right)$ is defined by
	\begin{align}
	\label{dQLCT}
	\mathcal{L}_{A_1,A_2}[f](\mathbf{w})=
\int_{\mathbb{R}^2}K_{A_1}^{{i}}(x_1,\omega_1)f(\mathbf{x})K_{A_2}^{{j}}(x_2,\omega_2)\rm{d}\mathbf{x},
	\end{align}	
	where $\mathbf{w}=(\omega_1,\omega_2)\in \mathbb{R}^{2}$ is regarded as the QLCT domain, and  the kernel signals $K_{A_1}^{{i}}(x_1,\omega_1)$, $K_{A_2}^{{j}}(x_2,\omega_2)$ are respectively given by
	\begin{align}
		\begin{split}
			K_{A_1}^{{i}}(x_1,\omega_1):=\begin{cases}
		\frac{1}{\sqrt{2\pi i b_1}}e^{{i}\left(\frac{a_1}{2b_1}x_1^2-\frac{x_1\omega_1}{b_1}+\frac{d_1}{2b_1}\omega_1^2 \right) },   &b_1\neq0  \\
		\sqrt{d_1}e^{{i}\frac{c_1d_1}{2}\omega_1^2}\delta(x_{1}-d_{1}w_{1}),    &b_1=0
		\end{cases}
		\end{split}
	\end{align}
	and
	\begin{align}
		\begin{split}
			K_{A_2}^{{j}}(x_2,\omega_2):=\begin{cases}
		\frac{1}{\sqrt{2\pi j b_2}}e^{{j}\left(\frac{a_2}{2b_2}x_2^2-\frac{x_2\omega_2}{b_2}+\frac{d_2}{2b_2}\omega_2^2 \right) },   &b_2\neq0  \\
		\sqrt{d_2}e^{{j}\frac{c_2d_2}{2}\omega_2^2}\delta(x_{2}-d_{2}w_{2}),    &b_2=0
		\end{cases}
		\end{split}
	\end{align}
where $\delta(x)$ representing the Dirac function.
\end{definition}
Here we note that for $b_s=0,s=1,2$ the QLCT of a signal boils down to chirp multiplication operations, and it is of no particular interest for our objective in this work. So without loss of generality, we set $b_s\neq0$ in rest of paper.

\begin{lemma} Suppose $f\in L^2\left( \mathbb{R}^2,\mathbb{H}\right)$, then the inversion of the QLCT of $f$ is given by
	\begin{align}
		\begin{split}
		\label{eIQLCT}
		f(\mathbf{x})&=\mathcal{L}_{A_1,A_2}^{-1}[\mathcal{L}_{A_1,A_2}[f]](\mathbf{x})\\
&=\int_{\mathbb{R}^2}K_{A_1}^{-{i}}(x_1,\omega_1)
\mathcal{L}_{A_1,A_2}\left\{f\right\}(\mathbf{w})K_{A_2}^{-{j}}(x_2,\omega_2)\rm{d}\mathbf{w}.
		\end{split}
	\end{align}	
\end{lemma}

\parindent=0mm \vspace{.1in}

 We now generalize the definitions of \cite{HI14,HS13} as follows:

\parindent=0mm \vspace{.1in}

\begin{definition}{\bf(QOLCT.)}
Let $A_s =\left[\begin{array}{cccc}a_s & b_s &| & p_s\\c_s & d_s &| & q_s \\\end{array}\right] $, be a matrix parameter such that $a_s$, $b_s$, $c_s$, $d_s$, $p_s$, $q_s \in \mathbb R,\quad b_s\ne 0$ and $ a_sd_s-b_sc_s=1,$ for $s=1,2.$ The two-sided quaternion offset linear canonical transform of any quaternion valued function $f\in L^2(\mathbb R^2,\mathbb H)$, is given by
\begin{equation}\label{dQOLCT}
\mathcal O_{A_1,A_2}^{i,j}\big[f(\mathbf x)\big](\mathbf w)= \int_{\mathbb R^2} K_{A_1}^i(x_1,w_1) \,f(\mathbf x)\,K_{A_2}^j(x_2,w_2)d\mathbf x
\end{equation}
where $\mathbf x=(x_{1},x_{2}),\, w=(w_{1},w_{2})$ and the kernel signals $K_{A_1}^{{i}}(x_1,w_1)$ and $K_{A_2}^j(x_2,w_2)$ are respectively given by
\begin{equation}\label{eqn k1}
K_{A_1}^{{i}}(x_1,w_1)=\dfrac{1}{\sqrt{2\pi b_1{i}}} \,e^{\frac{{i}}{2b_1}\big[a_1x_1^2-2x_1(w_1-p_1)-2w_1(d_1p_1-b_1q_1)+d_1(w_1^2+p_1^2)\big]},b_1\neq0\,
\end{equation}

\begin{equation}\label{eqn k2}
K_{A_2}^{{j}}(x_2,w_2)= \dfrac{1}{\sqrt{2\pi b_2{j}}} \,e^{\frac{{j}}{2b_2}\big[a_2x_2^2-2x_2(w_2-p_2)-2w_2(d_2p_2-b_2q_2)+d_2(w_2^2+p_2^2)\big]},\,b_2\neq0\,
\end{equation}
\end{definition}

 {\bf Note:}The left-sided and right-sided QOLCT can be defined correspondingly by placing the two above kernels both on the left or on the right, respectively.
\parindent=0mm \vspace{.1in}

 \begin{lemma} Suppose $f\in L^2(\mathbb R^2,\mathbb H)$ then the inversion of two-sided QOLCT is given by
\begin{align*}\label{IQOLCT}
f(\mathbf x) &=\int_{\mathbb R^2}{K_{A_1}^{-i}(x_1,w_1)} \, \mathcal O_{A_1,A_2}^{i,j}\big[f\big](\mathbf w)\,{K_{A_2}^{-j}(x_2,w_2)} \,d\mathbf w.
\end{align*}
\end{lemma}

\parindent=0mm \vspace{.1in}

 \begin{lemma}(Plancherel for QOLCT)  Every two dimensional quaternion valued function $f \in L^2(\mathbb R^2,\mathbb H)$ and its two-sided QOLCT are related to the Plancherel identity in the following way:
\begin{align}
\big\|\mathcal O_{A_1,A_2}^{i,j}\big[f\big]\big\|_{2}=\left\|f\right\|_{2}.\,
\end{align}
\end{lemma}

\parindent=0mm \vspace{.1in}

\section{\bf  Short-time Quaternion Offset Linear Canonical Transform(ST-QOLCT)}

In this section, we shall formally introduce the notion of the two-sided  Short-time quaternion offset linear canonical transform (ST-QOLCT)  then establish some properties of the proposed transform.

\parindent=0mm \vspace{.1in}
\begin{definition}\label{ST-QOLCT}{\bf(ST-QOLCT.)} Let $A_s =\left[\begin{array}{cccc}a_s & b_s &| & p_s\\c_s & d_s &| & q_s \\\end{array}\right] $, be a matrix parameter such that $a_s$, $b_s$, $c_s$, $d_s$, $p_s$, $q_s \in \mathbb R,\quad b_s\ne 0$ and $ a_sd_s-b_sc_s=1,$ for $s=1,2.$ The two-sided short-time quaternion offset linear canonical transform of any quaternion valued function $f\in L^2(\mathbb R^2,\mathbb H)$, with respect window function $\phi\in  L^2(\mathbb R^2,\mathbb H)$ is given by  \begin{equation}\label{eqn ST-QOLCT}
\mathcal{S}^{\mathbb H}_{\phi,A_1,A_2}\big[f\big](\mathbf{w,u})= \int_{\mathbb R^2} K_{A_1}^{\bf{i}}(x_1,w_1) \,f(\mathbf x)\overline{\phi(\mathbf{x-u})}\,K_{A_2}^{{j}}(x_2,w_2)dt
\end{equation}
where $\mathbf x=(x_{1},x_{2}),\, \mathbf w=(w_{1},w_{2}),\, \mathbf u= (u_1,u_2)$ and the quaternion kernels $K_{A_1}^i(x_1,w_1)$ and $K_{A_2}^j(x_2,w_2)$ are given by equation 2.23 and 2.24 respectively.
\end{definition}
\parindent=0mm \vspace{.1in}
{\bf Note:} It is worth to note that the quaternion ST-OLCT(\ref{ST-QOLCT}), boils down to various linear integral transforms such as:
\begin{itemize}
\item Short-time versions of quaternion linear canonical transform when matrices parameters $A_s =\left[\begin{array}{cccc}a_s & b_s &| & 0\\c_s & d_s &| & 0 \\\end{array}\right] $,\\

\item Quaternion short-time fractional Fourier transform when $A_s =\left[\begin{array}{cccc}\cos\theta & \sin\theta &| & 0\\-\sin\theta & \cos\theta &| & 0 \\\end{array}\right] $,\\

\item Quaternion short-time Fourier transform when $A_s =\left[\begin{array}{cccc}1 & 0 &| & 0\\0 & 1 &| & 0 \\\end{array}\right]$.

\end{itemize}
\parindent=0mm \vspace{.1in}

First of all we define the relation between ST-QOLCT and QOLCT, we begin as:\\
Since $b_s\neq 0,\, s=1,2,$ as in other cases proposed transform reduces to a chrip multiplications. Thus for fixed $\mathbf u$ we have

\begin{align}\label{IST1}
\mathcal S^{\mathbb H}_{\phi,A_1,A_2}\big[f\big](\mathbf{w,u})&=\mathcal O_{A_1,A_2}^{i,j}\big[f(\mathbf x)\overline{ \phi(\mathbf{x-u})}\big](\mathbf w)\\
&=\int_{\mathbb R^2} K_{A_1}^i(x_1,w_1) \,f(\mathbf x)\overline{\phi(\mathbf{x-u})}\,K_{A_2}^j(x_2,w_2)d\mathbf x
\end{align}

\parindent=0mm \vspace{.1in}

Applying the inverse QOLCT to (\ref{IST1}), we have\\
\begin{align}\label{IST2}
f\overline{\Theta}_{\mathbf u}(\mathbf x)=f(\mathbf x)\overline{\phi(\mathbf{x-u})}&=\{\mathcal O_{A_1,A_2}^{i,j}\}^{-1}\big[\mathcal{S}^{\mathbb H}_{\phi,A_1,A_2}\big[f\big](\mathbf{w,u}) \big](\mathbf x)\\
&=\int_{\mathbb R^2} K_{A_1}^{-i}(x_1,w_1) \mathcal{S}^{\mathbb H}_{\phi,A_1,A_2}\big[f\big](\mathbf {w,u})K_{A_2}^{-j}(x_2,w_2)d\mathbf w
\end{align}
where $f\overline{\Theta}_{\mathbf u}(\mathbf x)$ is known as modified signal.\\

Next we give the relation between two-sided ST-QOLCT and two-sided QFT, for that we have following lemma. It will be useful for our analysis of the ST-QOLCT.

\begin{lemma}\label{QFT_STQOLCT}
The two-sided ST-QOLCT(\ref{ST-QOLCT}) of a signal $f \in L^2({\mathbb R}^2, {\mathbb H})$\ can be reduced to the two-sided QFT(\ref{QFT}) as
\begin{align*}\label{lem qft}
\mathcal {S}^{\mathbb H}_{\phi,A_1,A_2}\big[f\big](w,u)= \frac{1}{\sqrt{2\pi {i  b}_1}}e^{i[-\frac{1}{b_1}{w_1(d}_1p_1-b_1s_1)+\frac{d_1}{2b_1}{(w}^2_1+{p }^2_1)]}\mathcal F^{i,j}(h)\left(\frac{\mathbf w}{\mathbf b},\mathbf u\right)\\
\times \frac{1}{\sqrt{2\pi {j  b}_2}}e^{j[-\frac{1}{b_2}{w_2(d}_2p_2-b_2s_2)+\frac{d_2}{2b_2}{(w}^2_2+{p }^2_2)]}
\end{align*}
where
\begin{equation}\label{function_h}
h(\mathbf{x,u})= e^{i[\frac{a_1}{2b_1}x^2_1+\frac{1}{b_1}x_1p_1]}f\overline{\Theta}_{\mathbf u}(\mathbf x)e^{j[\frac{a_2}{2b_2}x^2_2+\frac{1}{b_2}x_2p_2]}
\end{equation}
and  $\mathbf b=(b_1,b_2)$, $\mathcal F^{i,j}(h)$ is the QFT of signal $h$ given by (\ref{QFT}).
\end{lemma}
\begin{proof}
From the definition of the ST-QOLCT, we have
$$\begin{array}{lcr}
\mathcal S^{\mathbb H}_{\phi,A_1,A_2}\big[f\big](\mathbf{w,u})&&\\\\
\qquad=\displaystyle\int_{\mathbb R^2} K_{A_1}^{{i}}(x_1,w_1) f(x)\overline{\phi(x-u)}\,K_{A_2}^{{j}}(x_2,w_2)dx&&\\\\
\qquad=\displaystyle\int_{\mathbb R^2}\dfrac{1}{\sqrt{2\pi b_1{i}}} e^{\frac{{i}}{2b_1}\big[a_1x_1^2-2x_1(w_1-p_1)-2w_1(d_1p_1-b_1q_1)+d_1(w_1^2+p_1^2)\big]}f(x)\overline{\phi(x-u)}&&\\\\
\qquad\qquad\qquad\times e^{\frac{{j}}{2b_2}\big[a_2x_2^2-2x_2(w_2-p_2)-2w_2(d_2p_2-b_2q_2)+d_2(w_2^2+p_2^2)\big]}dx&&\\\\
\qquad=\frac{1}{\sqrt{2\pi {i  b}_1}}e^{i[-\frac{1}{b_1}{w_1(d}_1p_1-b_1s_1)+\frac{d_1}{2b_1}{(w}^2_1+{p }^2_1)]}\displaystyle\int_{\mathbb R^2}e^{-i\frac{1}{b_1}x_1w_1}\left(e^{i[\frac{a_1}{2b_1}x^2_1+\frac{1}{b_1}x_1p_1]}f(x)\overline{\phi(\mathbf {x-u})}\right.&&\\\\
\qquad\qquad\qquad\times \left.e^{j[\frac{a_2}{2b_2}x^2_2+\frac{1}{b_2}x_2p_2]}\right)e^{-j\frac{1}{b_2}x_2w_2}dx\frac{1}{\sqrt{2\pi {j  b}_2}}e^{j[-\frac{1}{b_2}{w_2(d}_2p_2-b_2s_2)+\frac{d_2}{2b_2}{(w}^2_2+{p }^2_2)]}\\\\
\end{array}$$

on setting
$
h(\mathbf{x,u})= e^{i[\frac{a_1}{2b_1}x^2_1+\frac{1}{b_1}x_1p_1]}f\overline{\Theta}_{\mathbf u}(\mathbf x)e^{j[\frac{a_2}{2b_2}x^2_2+\frac{1}{b_2}x_2p_2]}$
in the above equationwe get the desired result.
\begin{align*}
\mathbb S^{\mathbb H}_{\phi,A_1,A_2}\big[f\big](\mathbf{w,u})= \frac{1}{\sqrt{2\pi {i  b}_1}}e^{i[-\frac{1}{b_1}{w_1(d}_1p_1-b_1s_1)+\frac{d_1}{2b_1}{(w}^2_1+{p }^2_1)]}\mathcal F^{i,j}(h)\left(\frac{\mathbf w}{\mathbf b},\mathbf u\right)\\
\times\frac{1}{\sqrt{2\pi {j  b}_2}}e^{j[-\frac{1}{b_2}{w_2(d}_2p_2-b_2s_2)+\frac{d_2}{2b_2}{(w}^2_2+{p }^2_2)]}
\end{align*}
\end{proof}
\subsection{Some properties of ST-QOLCT}

\begin{theorem}\label{linear bounded}
Let $f,\phi\in L^2({\mathbb R}^2,{\mathbb H})$. Then its ST-QOLCT satisfies:

 (i)The map $f \longrightarrow\mathcal{S}^{\mathbb H}_{\phi,A_1,A_2}[f]$ \ is real linear.\\
 (ii)$\mathcal S^{\mathbb H}_{\phi,A_1,A_2}[f]$ is is uniformly continuous and bounded on the time–frequency plane $\mathbb R^2\times\mathbb R^2$ and satisfies :
 \begin{align}
		\begin{split}
		|\mathcal{S}^{\mathbb H}_{\phi,A_1,A_2}\big[f\big](\mathbf{w,u})|\leq \frac{1}{2\pi\sqrt{|b_{1}b_{2}|}} \|f\|_{L^{2}\mathbb({R}^2,H)} \|\phi\|_{L^{2}(\mathbb{R}^2,\mathbb H)}\\
\end{split}
	\end{align}
\begin{proof}
The proof of (i) follows  by definition(\ref{ST-QOLCT}) and (ii) is proved in (\cite{OWN1}, Thm 1).
\end{proof}
\end{theorem}
\begin{theorem}{\bf(Moyal’s formula).}\label{myl} Let $\phi,\psi \in L^2(\mathbb{R}^2,\mathbb{H})$ be a fixed non-zero window functions and $f,g \in L^2\left( \mathbb{R}^2,\mathbb{H}\right)$
then
\begin{equation}
\left\langle \mathcal S^{\mathbb H}_{\phi,A_1,A_2}[f](\mathbf{w,u}), {\mathcal S^{\mathbb H}_{\psi,A_1,A_2}[g](\mathbf{w,u})} \right\rangle
=[\left\langle f,g\right\rangle\left\langle \phi ,\psi\right\rangle]
\end{equation}
\begin{proof}
The proof is already present in \cite{OWN1}.
\end{proof}
\end{theorem}

Consequences of theorem \ref{myl}.

\parindent=8mm \vspace{.2in}

(i) If $\phi=\psi$, then
\begin{equation}\label{a}
\langle \mathcal S^{\mathbb H}_{\phi,A_1,A_2}[f](\mathbf{w,u}), {\mathcal S^{\mathbb H}_{\phi,A_1,A_2}[f_2](\mathbf{w,u})} \rangle
=\|\phi\|^2_{L^2(\mathbb R^2,\mathbb H)}\langle f,g\rangle\end{equation}

(ii) If $f=g$, then
\begin{equation}\label{b}
\langle \mathcal S^{\mathbb H}_{\phi,A_1,A_2}[f](\mathbf{w,u}), {\mathcal S^{\mathbb H}_{\psi,A_1,A_2}[f](\mathbf{w,u})} \rangle
=\langle{\phi,\psi}\rangle\|f\|^{2}_{L^2(\mathbb{R}^2,\mathbb H)}
\end{equation}
(iii) If  $f=g$ and $\phi=\psi$, then
\begin{equation}\label{c}
\langle \mathcal S^{\mathbb H}_{\phi,A_1,A_2}[f](\mathbf{w,u}), {\mathcal S^{\mathbb H}_{\phi,A_1,A_2}[f](\mathbf{w,u})} \rangle=
\int_{\mathbb{R}^2}\int_{\mathbb{R}^2}| \mathcal S^{\mathbb H}_{\phi,A_1,A_2}[f](\mathbf{w,u})|^{2}dudw
=\|\phi\|^2_{L^2(\mathbb R^2,\mathbb H)}\|f\|^{2}_{L^2(\mathbb{R}^2,\mathbb H)} \end{equation}

Note (\ref{c}) is known as the energy preserving relation for the proposed
ST-QOLCT.
\parindent=8mm \vspace{.2in}

\begin{remark}{\bf(Isometry )}. For  $\|\phi\|^2_{L^2(\mathbb R^2,\mathbb H)}=1 $  (\ref{c}) reduces to
\begin{equation}\int_{\mathbb{R}^2}\int_{\mathbb{R}^2}|\mathcal S_\phi^{A_1,A_2}[f](\mathbf{w,u})|^2 dudw =\|f\|^2_{L^{2}(\mathbb{R}^2,\mathbb{H})}\end{equation}
i.e the proposed ST-QOLCT(\ref{ST-QOLCT}) becomes an
isometry from ${L^{2}(\mathbb{R}^2,\mathbb{H})}$ into ${L^{2}(\mathbb{R}^2,\mathbb{H})}$.In other words, the total energy of a quaternion-valued signal computed in the
in the quaternion short-time offset linear canonical  domain is equal to the
total energy computed in the spatial domain.
\end{remark}
\begin{theorem}\label{reconstruction}{\bf(Reconstruction formula).} Every 2D quaternion signal $f\in L^{2}(\mathbb{R}^2,\mathbb{H})$
 can be fully reconstructed by the formula
\begin{equation}
f(\mathbf x)=\frac{1}{\|\phi\|^2} \int_{\mathbb R^2}\int_{\mathbb R^2} K_{A_1}^{-i}(x_1,w_1)\mathcal S^{\mathbb H}_{\phi, A_1,A_2} [f](\mathbf{w,u}){ K}^{-j}_{A_2}(x_2,w_2) \phi(\mathbf{x-u})d\mathbf w d\mathbf u.
\end{equation}
\begin{proof} Already proved in \cite{OWN1}\end{proof}
\end{theorem}

\section{\bf Uncertainty principles for the QWLCT}
In this section we study several different kinds of uncertainty principles associated with ST-QOLCT.
 \subsection{Donoho-Stark's uncertainty principle}
In this subsection, according to the relationship between the ST-QOLCT and the QFT, we present a exquisite  uncertainty principle on $\mathbb R^2$ concerning to the Donoho-Stark's uncertainty principle. First we revisit the concept of $\epsilon-concentrate$ of a quaternion valued signal  on a measurable set $M\subseteq\mathbb R^2,$. Let us begin with the following definition.
\begin{definition}\label{econ}
For $\epsilon\ge 0,$ a quaternion valued signal $f\in L^2(\mathbb R^2,\mathbb H)  $ is said to be $\epsilon-concentrated$ on a measurable set $M\subseteq\mathbb R^2,$ if
\begin{equation}\label{eqn econ}
\left(\int_{\mathbb R^2\setminus M}|f(\mathbf x)|^2d\mathbf x\right)^{\frac{1}{2}}\le\epsilon\|f\|_2
\end{equation}
If  $0\le\epsilon\le\frac{1}{2},$ then the most of energy is concentrated on $M$, and  $M$ is indeed the essential support of $f$, if $\epsilon= 0,$ then $D$
is the exact support of $f$.
Similarly, we say that its $\mathcal F^{i,j}$ is $\epsilon-concentrated$ on a measurable set $N\subseteq \mathbb R^2,$ if
\begin{equation}\label{eqn econft}
\left(\int_{\mathbb R^2 \setminus N}|\mathcal F^{i,j}[f(\mathbf x)](\mathbf w)|^2d\mathbf w\right)^{\frac{1}{2}}\le\epsilon\|\mathcal F^{i,j}[f]\|_2
\end{equation}
\end{definition}
\begin{lemma}{\bf(Donoho-Stark's uncertainty principle for QFT\cite{32})}
\label{DONOqft}
Let $f\in L^2(\mathbb{R}^2,\mathbb{H})$ with $f\neq0$ is $\epsilon_{M}-$concentrated on $M \subseteq\mathbb{R}^2$,
and $F_{Q}(f)$ is $\epsilon_{N}-$concentrated on $N \subseteq\mathbb{R}^2$. Then
\begin{align}
		\begin{split}
|M||N|\geq2\pi (1-\epsilon_{M}-\epsilon_{N})^{2}.
		\end{split}
	\end{align}
where $|M|$ and $|N|$ are the measures of the sets $M$ and $N$.
\end{lemma}
\begin{definition}
Let $f,\phi\in L^2(\mathbb R^2,\mathbb H)$ where $\phi $ is a non zero window function then $\mathcal S^{\mathbb H}_{\phi,A_1,A_2}[f](\mathbf{w,u})$ is $\epsilon_{N}-$concentrated on $N \subseteq\mathbb{R}^2,$ if

\begin{equation}\left(\int_{\mathbb R^2}\int_{\mathbb R^2\setminus N}|\mathcal S^{\mathbb H}_{\phi,A_1,A_2}[f(\mathbf x)](\mathbf {w,u})|^2d\mathbf w d\mathbf u\right)^{\frac{1}{2}}\le\epsilon_N\|\mathcal S^{\mathbb H}_{\phi,A_1,A_2}[f]\|_2\end{equation}

where $\|\mathcal S^{\mathbb H}_{\phi,A_1,A_2}[f]\|_2=\left(\int_{\mathbb R^2}\int_{\mathbb R^2}|\mathcal S^{\mathbb H}_{\phi,A_1,A_2}[f(\mathbf x)](\mathbf {w,u})|^2d\mathbf w d\mathbf u\right)^{\frac{1}{2}}$
\end{definition}
\begin{theorem}\label{DONO ST-QOLCT}{\bf(Donoho-Stark's uncertainty principle for ST-QOLCT).}
 Let $\phi$ be the nonzero quaternion  window function and $f\ne 0$ be a quaternion signal function in $L^2(\mathbb{R}^2,\mathbb{H})$
is $\epsilon_{M}-$concentrated on measurable set $M \subseteq\mathbb{R}^2$, and $\mathcal{S}^{\mathbb H}_{\phi,A_1,A_2}[f](\mathbf{w,u})$ is
$\epsilon_{N}-$concentrated on $N \subseteq\mathbb{R}^2$.
Then
\begin{align}
		\begin{split}
|M||N|\geq2\pi b_1b_2 (1-\epsilon_{M}-\epsilon_{N})^{2}.
		\end{split}
	\end{align}
\begin{proof}
We have from Lemma(\ref{QFT_STQOLCT})
\begin{align*}
\mathbb S^{\mathbb H}_{\phi,A_1,A_2}\big[f\big](\mathbf{w,u})= \frac{1}{\sqrt{2\pi {i  b}_1}}e^{i[-\frac{1}{b_1}{w_1(d}_1p_1-b_1s_1)+\frac{d_1}{2b_1}{(w}^2_1+{p }^2_1)]}\mathcal F^{i,j}(h)\left(\frac{\mathbf w}{\mathbf b},\mathbf u\right)\\
\times\frac{1}{\sqrt{2\pi {j  b}_2}}e^{j[-\frac{1}{b_2}{w_2(d}_2p_2-b_2s_2)+\frac{d_2}{2b_2}{(w}^2_2+{p }^2_2)]}
\end{align*}
where
\begin{equation}\label{function_h}
h(\mathbf{x,u})= e^{i[\frac{a_1}{2b_1}x^2_1+\frac{1}{b_1}x_1p_1]}f\overline{\Theta}_{\mathbf u}(\mathbf x)e^{j[\frac{a_2}{2b_2}x^2_2+\frac{1}{b_2}x_2p_2]}
\end{equation}

For $\mathbf {u=x},$ we have $|h(\mathbf x)|=|f(\mathbf x)|\overline{\phi(0)}|$ , with $|\overline{\phi(0)}|> 0,$ since $f$ is $\epsilon_{M}-$concentrated on measurable set $M \subseteq\mathbb{R}^2$, then by definition(\ref{econ})
$$\left(\int_{\mathbb R^2\setminus M}|f(\mathbf x)|^2d\mathbf x\right)^{\frac{1}{2}}\le\epsilon_M\|f\|_2\Rightarrow\left(\int_{\mathbb R^2\setminus M}|h(\mathbf x)|^2d\mathbf x\right)^{\frac{1}{2}}\le\epsilon_M\|f\|_2$$
 i.e. $h(\mathbf x)$ is $\epsilon_{M}-$concentrated on measurable set $M \subseteq\mathbb{R}^2$.\\\\

 Also  it is given that $\mathcal{S}^{\mathbb H}_{\phi,A_1,A_2}[f](\mathbf{w,u})$ is
$\epsilon_{N}-$concentrated on $N \subseteq\mathbb{R}^2$ and we have  $|\mathbb S^{\mathbb H}_{\phi,A_1,A_2}\big[f\big](\mathbf{w,u})|=|\frac{1}{\sqrt{2\pi\mathbf b_1}}\mathcal F^{i,j}(h)\left(\frac{\mathbf w}{\mathbf b},\mathbf u\right)\frac{1}{\sqrt{2\pi\mathbf b_1}}|$, which implies $\mathcal F^{i,j}(h)\left(\frac{\mathbf w}{\mathbf b},\mathbf u\right)$ is $\epsilon_{N}-$concentrated on $N \subseteq\mathbb{R}^2$, that is to say, is $\mathcal F^{i,j}(h)\left({\mathbf w},\mathbf u\right)$ is $\epsilon_{N}-$concentrated on $\frac{N}{\mathbf b} \subseteq\mathbb{R}^2$.
 Hence, applying Lemma (\ref{DONOqft}) to the function $h$, we obtain
 \begin{equation}|M|\left|\frac{N}{\mathbf b}\right|\geq2\pi (1-\epsilon_{M}-\epsilon_{N})^{2}.\end{equation}
 Which gives
 \begin{equation}|M||{N}|\geq2\pi  b_1b_2 (1-\epsilon_{M}-\epsilon_{N})^{2}.\end{equation}
\end{proof}
which completes the proof.
\end{theorem}
\begin{corollary} If $f\overline{\Theta}_{\mathbf u}(\mathbf x)\in L^2(\mathbb R^2,\mathbb H),$  $supp f\overline{\Theta}_{\mathbf u}(\mathbf x)\subseteq M$  and $supp \mathbb S^{\mathbb H}_{\phi,A_1,A_2}\big[f\big](\mathbf{w,u})\subseteq N,$ then
 \begin{equation}|M||{N}|\geq2\pi \mathbf b .\end{equation}
 \begin{proof}
It is clear from definition (\ref{econ}) that $f(\mathbf x)$ is $0-$concentrated on $M$ iff $supp(f)=M.$ Therefore if we take $\epsilon_{M}=\epsilon_{N}=0$ in theorem (\ref{DONO ST-QOLCT}) we get desired result.
 \end{proof}
\end{corollary}
\parindent=0mm \vspace{.1in}

\subsection{Hardy's uncertainty principle}
G.H Hardy introduced  Hardy's uncertainty principle\cite{HA33} in 1933 which is  qualitative in nature, it states that it is impossible for a non zero signal function and its Fourier transform to decrease very rapidly simultaneously. We first present
the Hardy's UP for the Two-sided QFT\cite{17}.
\begin{lemma}\label{Hardy QFT}
Let $\alpha \ $  and $\beta $  are positive constants .Suppose $f\in L^{2}({\mathbb R}^2,{\mathbb H})$  with

\begin{equation}{|f\left(\mathbf x\right)|}\le {Ce}^{-\alpha {\left|\mathbf x\right|}^2},\  \mathbf x\in {\mathbb R}^2.\end{equation}

\begin{equation}{|\ {\mathcal F}^{i,j }\left\{f\right\}\left(\mathbf w\right)|}\le {C' e}^{-\beta {\left|\mathbf w\right|}^2},\ \mathbf w\in {\mathbb R}^2.\end{equation}

for some positive constants $C,C'.$Then, three cases can occur :
\begin{itemize}
 \item   If $\alpha \beta > \frac {1}{4}$, then $f=0$.
 \item  If  $\alpha \beta = \frac {1}{4}$, then $\ f(t)=Ae^{-\alpha {\left|\mathbf x\right|}^2}$,\ whit $A$\ is a quaternion constant.
 \item  If $\alpha \beta <\frac {1}{4},$\ then there are infinitely many such functions $f$.
\end{itemize}
\end{lemma}

By using Lemma\ref{QFT_STQOLCT} and Lemma \ref{Hardy QFT} , we derive Hardy’s uncertainty principle for the ST-QOLCT..
\begin{theorem}\label{ST-QOLCT_Hardy}
Let $\phi\in L^2(\mathbb R^2,\mathbb H)$  be a non zero window function . Suppose $f\in L^{2}({\mathbb R}^2,{\mathbb H})$  with
\begin{equation}\label{Hardy1}{\left|f(\mathbf x)\right|}\le {Ce}^{-\alpha {\left|\mathbf x\right|}^2},  \mathbf x\in {\mathbb R}^2. \end{equation}
\begin{equation}\label{Hardy2}{\left|\mathcal S^{\mathbb H}_{\phi,A_1,A_2}f(\mathbf b \mathbf{w+p},\mathbf u)\right|} \le C'e^{-\beta {\left|\mathbf w\right|}^2},\ \mathbf w\in {\mathbb R}^2. \end{equation}
for some constants $\alpha, \beta>0$ and $C,C'$ are positive constants ,then:\\
 \begin{itemize}
 \item  If $\alpha \beta >\frac{1}{4}$, then $f=0$.
 \item  If $\alpha \beta =\frac{1}{4}$, then $\ f(\mathbf x)=e^{-i\frac{a_1}{2b_1}x^2_1-i\frac{1}{b_1}x_1p_1}\frac{A}{\overline{(\phi(0))}} e^{-\alpha {\left|\mathbf x\right|}^2}e^{-j\frac{a_2}{2b_2}x^2_2-j\frac{1}{b_2}x_2j_2}$, where $A$\  is a quaternion constant.
 \item  If  $\alpha \beta <\frac{1}{4},$\ then there are infinitely many $f$.
\end{itemize}
\begin{proof}
On substituting $\mathbf{u=x}$ in (\ref{function_h}),we have
\begin{align*}
h(\mathbf{x})= e^{i[\frac{a_1}{2b_1}x^2_1+\frac{1}{b_1}x_1p_1]}f(\mathbf x)\overline{\phi(0)}e^{j[\frac{a_2}{2b_2}x^2_2+\frac{1}{b_2}x_2p_2]}
\end{align*}
clearly RHS of above equation belongs to $L^2(\mathbb R^2,\mathbb H)$  and $|\overline{\phi(0)}|$ is a positve quantity and
$$|h(\mathbf{x})|=|f(\mathbf x)||\overline{\phi(0)}|\le|\overline{\phi(0)}|{Ce}^{-\alpha {\left|x\right|}^2}=C_1{e}^{-\alpha {\left|\mathbf x\right|}^2}\eqno(4.3)$$
Now applying(\ref{QFT_STQOLCT}) and (\ref{Hardy2}),we have
\begin{equation}
|\mathcal F ^{i,j}[h(\mathbf x)](\mathbf w)|=\sqrt{ b_1b_2}{\left|\mathcal S^{\mathbb H}_{\phi,A_1,A_2}f(\mathbf{b w}+p,\mathbf u)\right|} \le\sqrt{\mathbf b}C_0e^{-\beta|\mathbf w|^2}
\end{equation}
Therefore, it follows from Lemma\ref{Hardy QFT} that,\\
If $\alpha \beta >\frac{1}{4}$ then $h=0$, so $f=0$.\\
If $\alpha \beta =\frac{1}{4}$ then\\
$f\left(\mathbf x\right)=Ae^{-\alpha {\left|\mathbf x\right|}^2}$, for some  constant $A.$\\
Hence
\[f\left(t\right)=e^{-i\frac{a_1}{2b_1}x^2_1-i\frac{1}{b_1}x_1p_1}\frac{A}{\overline{(\phi(0))}} e^{-\alpha {\left|\mathbf x\right|}^2}e^{-j\frac{a_2}{2b_2}x^2_2-j\frac{1}{b_2}x_2j_2}.\]

If $\alpha \beta <\frac{1}{4},$ then there are infinitely many such functions $f$, that verify  \eqref{Hardy1} and ( \eqref{Hardy2}.\\
This completes the proof.

\end{proof}
\end{theorem}
It follows from theorem \ref{ST-QOLCT_Hardy} that it is impossible for a signal $f$\ and its two-sided ST-QOLCT to both decrease very rapidly.

\subsection{Beurling's uncertainty principle}

Beurling's uncertainty principle \cite{BE89}, \cite{HO91} is a mutant of Hardy's uncertainty principle.
The following Lemma is the Beurling's  uncertainty principle for the Two-sided QOLCT (\cite{gen} Cor. 4.7)
.
\begin{lemma}\label{beu lem QOLCT}
Let $f \in L^2\left(\mathbb {\mathbb R}^2,\mathbb H\right)\ and\ \ d\ge 0~$   satisfy
\begin{equation}
\int_{{\mathbb R}^2}{\int_{{\mathbb R}^2}{ \frac{{\left| f(\mathbf x )\right|}\ {\left|{\mathcal O}^{i,j }_{A_1,A_2}\left[f\right]\left(\mathbf w\right)\right|}}{{(1+\left|\mathbf x\right|+\left|\mathbf w\right|)}^d}e^{\left|\mathbf x\right|\left|\mathbf w\right|}}}\ d\mathbf xd\mathbf w <\infty, \end{equation}
Then \\
$f\left(\mathbf x\right)=P(\mathbf x)e^{-a|\mathbf x|^2},$   a.e.\\

Where $a>0$ and $P$\ is a quaternion polynomial of degree $< \frac{d-2}{2}$. \\
In particular, $f=0$\  a.e. when $d\le 2.$
\end{lemma}

On the basis of lemma \ref{beu lem QOLCT}, we give the Beurlings' uncertainty principle associated with ST-QOLCT domains.
\begin{theorem}\label{bucp ST-QOLCT}
Let   $\phi,f \in L^2\left(\mathbb {\mathbb R}^2,\mathbb H\right)$ where $\phi$ be a non zero quaternion window function and $d\ge0$ satisfy
\begin{equation}\int_{{\mathbb R}^2}{\int_{{\mathbb R}^2}{ \frac{{\left| f(\mathbf x )\right||\overline{\phi\mathbf{(x-u)}}|}\ {\left|{\mathcal S}^{\mathbb H}_{\phi,A_1,A_2}\left[f\right](\mathbf{w,u})\right|}}{{(1+\left|\mathbf x\right|+\left|w\right|)}^d}e^{\left|\mathbf x\right|\left|\mathbf w\right|}}}\ dxdw<\infty\end{equation}
Then \\
$f\left(\mathbf x\right)=\frac{P(\mathbf x)}{\overline{\phi(\mathbf {x-u)}}}e^{-a|\mathbf x|^2},$   a.e.\\

Where $a>0$ and $P$\ is a quaternion polynomial of degree $< \frac{d-2}{2}$. \\
In particular, $f=0$\  a.e. when $d\le 2.$
\end{theorem}
\begin{proof}
From  (\ref{IST2}) we have $f\overline{\Theta}_\mathbf u(\mathbf x)=f(x)\overline{\phi(x-u)}\in L^2\left(\mathbb {\mathbb R}^2,\mathbb H\right),$ it follows that
 \begin{align*}
&\int_{{\mathbb R}^2}{\int_{{\mathbb R}^2}{ \frac{{\left| f\overline\Theta_\mathbf u(\mathbf x )\right|}\ {\left|{\mathcal O}^{i,j }_{A_1,A_2}\left[f\Theta_\mathbf u(\mathbf x )\right]\left(\mathbf w\right)\right|}}{{(1+\left|\mathbf x\right|+\left|\mathbf w\right|)}^d}e^{\left|\mathbf x\right|\left|\mathbf w\right|}}}\ d\mathbf xd\mathbf w\\
&=\int_{{\mathbb R}^2}{\int_{{\mathbb R}^2}{ \frac{{\left| f(\mathbf x )\right||\overline{\phi(\mathbf{x-u})}|}\ {\left|{\mathcal O}^{i,j }_{A_1,A_2}\left[f\overline\Theta_\mathbf u(\mathbf x)\right]\left(\mathbf w\right)\right|}}{{(1+\left|\mathbf x\right|+\left|\mathbf w\right|)}^d}e^{\left|\mathbf x\right|\left|\mathbf w\right|}}}\ d\mathbf xd\mathbf w\\
&=\int_{{\mathbb R}^2}{\int_{{\mathbb R}^2}{ \frac{{\left| f( \mathbf x )\right||\overline{\phi(\mathbf{x-u})}|}\ {\left|{\mathcal S}^{\mathbb H}_{\phi,A_1,A_2}\left[f\right]\left(\mathbf{w,u}\right)\right|}}{{(1+\left|\mathbf x\right|+\left|\mathbf w\right|)}^d}e^{\left|\mathbf x\right|\left|\mathbf w\right|}}}\ d\mathbf x\mathbf w<\infty.\end{align*}

Therefore by  Lemma \ref{beu lem QOLCT}, we have  $f\overline\Theta_\mathbf u(\mathbf x)=P(x)e^{-a|x|^2}$,  a.e where $a>0$ and $P$\ is a quaternion polynomial of degree $< \frac{d-2}{2}$. \\
i.e.\\
$f(\mathbf x)=\frac{P(\mathbf x)}{\overline{\phi(\mathbf{x-u})}}e^{-a|\mathbf x|^2}$\\
In particular, $f=0$\  a.e. when $d\le 2$  on account of $f\overline\Theta_\mathbf u(\mathbf x)=0$
\end{proof}
\subsection{Logarithmic uncertainty principle}
In this subsection we derive logarithmic uncertainty principle for ST-QOLCT
  by using Pitt's inequality for ST-QOLCT. Prior to that we derive Pitt's inequality for ST-QOLCT by using the Lemma \ref{QFT_STQOLCT} and Pitt's inequality for the QFT.
\begin{lemma}\label{Pitt}{\bf (Pitt's inequality for the two-sided QFT \cite{CKL15})}\\
 For  $f\in { \mathcal S}({\mathbb R}^2,{\mathbb H})$, and   $0\le \alpha <2$,
\begin{equation}\label{Pitt QFT}\int_{{\mathbb R}^2}{ {\left|\mathbf w\right|}^{-\alpha }}{\left\|{\mathcal F}^{i,j}\left\{f(\mathbf x)\right\}\left(\mathbf w\right)\right\|}^2d\mathbf w\  \le C_{\alpha } \int_{{\mathbb R}^2}{ {\left|\mathbf x\right|}^{\alpha }}{\left|f(\mathbf x)\right|}^2\ d\mathbf x.\end{equation}
With $C_{\alpha }:=\frac{{4\pi }^2}{2^{\alpha }}{{ [}\Gamma (\frac{2-\alpha }{4})/\Gamma (\frac{2+\alpha }4)]}^2$, and $\Gamma \left(.\right)$\ is the Gamma function and ${ \mathcal S}({\mathbb R}^2,{\mathbb H})$ denotes the
Schwartz space.
\end{lemma}
\begin{theorem}{\bf (Pitt's inequality of the ST-QOLCT.)}Under the assumptions of lemma \ref{Pitt}, we have
\begin{align}\label{pi}
		\begin{split}		\int_{\mathbb{R}^2}\int_{\mathbb{R}^2}|\mathbf{w}|^{-\alpha}|\mathcal S_{\phi,A_1,A_2}^{\mathbb H}[f](\mathbf{w,u})|^{2}\rm{d}\mathbf{u}\rm{d}\mathbf{w}
\leq \frac{1}{4\pi^{2}|{b_1b_2}|^{\alpha }}C_{\alpha}\|\phi\|^{2}_{L^2(\mathbb{R}^2,\mathbb{H})}
\int_{\mathbb{R}^2}|\mathbf{x}|^{\alpha}|f(\mathbf{x})|^{2}\rm{d}\mathbf{x},
		\end{split}
	\end{align}
\begin{proof}
By lemma (\ref{QFT_STQOLCT}), we have
\begin{align}\label{kk}
		\begin{split}
		\int_{\mathbb{R}^2}\int_{\mathbb{R}^2}|\mathbf{w}|^{-\alpha}|\mathcal S_{\phi,A_1,A_2}^{\mathbb H}[f](\mathbf{w,u})|^{2}\rm{d}\mathbf{u}\rm{d}\mathbf{w}&=
\frac{1}{4\pi^{2}|{b_1b_1}|}\int_{\mathbb{R}^2}\int_{\mathbb{R}^2}|\mathbf{w}|^{-\alpha}|\mathcal F^{i,j}(\textit{h})\left(\mathbf{\frac{w}{b},u}\right)|^{2}\rm{d}\mathbf{u}\rm{d}\mathbf{w}\\
&=\frac{1}{4\pi^{2}|{b_1b_1}|}\int_{\mathbb{R}^2}\int_{\mathbb{R}^2}|\mathbf{zb}|^{-\alpha}|\mathcal F^{i,j}(\textit{h})\left(\mathbf{z,u}\right)|^{2} |\mathbf b|\rm{d}\mathbf{u}\rm{d}\mathbf{z}\\
		\end{split}
	\end{align}
where last equation is obtained by taking $ \mathbf{zb}=\mathbf{w}.$ \\\\

Since $h(\mathbf{x,u})= e^{i[\frac{a_1}{2b_1}x^2_1+\frac{1}{b_1}x_1p_1]}f\overline\Theta_{\mathbf u}(\mathbf x)e^{j[\frac{a_2}{2b_2}x^2_2+\frac{1}{b_2}x_2p_2]}$
therefore by applying  Lemma \ref{Pitt}, we obtain from (\ref{kk})
\begin{align*}
		\begin{split}
	\int_{\mathbb{R}^2}\int_{\mathbb{R}^2}|\mathbf{w}|^{-\alpha}|\mathcal S_{\phi,A_1,A_2}^{\mathbb H}[f](\mathbf{w,u})|^{2}\rm{d}\mathbf{u}\rm{d}\mathbf{w}	&\leq \frac{1}{4\pi^{2}|{b_1b_2}|^{\alpha}}C_{\alpha}\int_{\mathbb{R}^2}\int_{\mathbb{R}^2}|\mathbf{x}|^{\alpha}|\textit{h}(\mathbf{x})|^{2}\rm{d}\mathbf{u}\rm{d}\mathbf{x}\\
&=\frac{1}{4\pi^{2}|{b_1b_2}|^{\alpha}}C_{\alpha}\int_{\mathbb{R}^2}\int_{\mathbb{R}^2}|\mathbf{x}|^{\alpha}|\textit{f}\overline\Theta_{\mathbf u}(\mathbf{x})|^{2}\rm{d}\mathbf{u}\rm{d}\mathbf{x}\\
&=\frac{1}{4\pi^{2}|{b_1b_2}|^{\alpha}}C_{\alpha}\int_{\mathbb{R}^2}\int_{\mathbb{R}^2}|\mathbf{x}|^{\alpha}|f(\mathbf x)\overline{\phi(\mathbf (x-u))}|^{2}\rm{d}\mathbf{u}\rm{d}\mathbf{x}\\
&=\frac{1}{4\pi^{2}|{b_1b_2}|^{\alpha}}C_{\alpha}\int_{\mathbb{R}^2}|\mathbf{x}|^{\alpha}|f(\mathbf x)|^2\int_{\mathbb{R}^2}|\overline{\phi(\mathbf (x-u))}|^{2}\rm{d}\mathbf{u}\rm{d}\mathbf{x}\\
&=\frac{1}{4\pi^{2}|{b_1b_2}|^{\alpha}}C_{\alpha}\|\phi\|^2\int_{\mathbb{R}^2}|\mathbf{x}|^{\alpha}|f(\mathbf x)|^2\rm{d}\mathbf{x}\\
		\end{split}
	\end{align*}

\end{proof}
\end{theorem}
\begin{theorem}{\bf(Logarithmic UP for the QOLCT)}\\
Let $f ,\phi\in {\mathcal S}\left({\mathbb R}^2,{\mathbb H}\right) $\ where $\phi$ is a non zero window function, then
\begin{align}\label{Becker2}
\int_{{\mathbb R}^2}\int_{{\mathbb R}^2}{\ {{ ln(} \left|\mathbf w\right|)\ \ }}{\left|{\mathcal S}^{\mathbb H}_{\phi,A_1,A_2}\left\{f\right\}\left(\mathbf{w,u}\right)\right|}^2d\mathbf w d\mathbf u+\frac{\|\phi\|^2}{4\pi^2}\int_{{\mathbb R}^2}{\ {\ln  \left(\left|\mathbf x\right|\right)\ }\ }{\left|f(\mathbf x)\right|}^2\ d\mathbf x\\
\ge \  \frac{(A+\ln|b_1b_2|)}{4\pi^2}\|\phi\|^2\int_{{\mathbb R}^2}{\ \ }{\left|f(\mathbf x)\right|}^2\ d\mathbf x,
\end{align}
with   $A={\ln  \left(2\right)}+{\Gamma '\left(\frac{1}{2}\right)}/{\Gamma (\frac{1}{2})}.$
\begin{proof}
Based on Pitt's inequality, Logarithmic uncertainty principle for the two sided ST-QOLCT can be proved by taking a  function $\Psi$ as
\begin{align*}
\Psi(\alpha)=\int_{\mathbb{R}^2}\int_{\mathbb{R}^2}|\mathbf{w}|^{-\alpha}|\mathcal S^{\mathbb H}_{\phi,A_1,A_2}[f](\mathbf{w,u})|^{2}\rm{d}\mathbf{u}\rm{d}\mathbf{w}-
\frac{D_\alpha}{|{b_1b_2}|^\alpha}\|\phi\|^{2}_{L^2(\mathbb{R}^2,\mathbb{H})}
\int_{\mathbb{R}^2}|\mathbf{x}|^{\alpha}|f(\mathbf{x})|^{2}\rm{d}\mathbf{x}
\end{align*}
 where $D_\alpha=\frac{C_\alpha}{4\pi^2}$ \\\\
  Implies
  \begin{align*}
\Psi'(\alpha)=\int_{\mathbb{R}^2}\int_{\mathbb{R}^2}|\mathbf{w}|^{-\alpha}\ln|\mathbf{w}||\mathcal S^{\mathbb H}_{\phi,A_1,A_2}[f](\mathbf{w,u})\|^{2}\rm{d}\mathbf{u}\rm{d}\mathbf{w}-
{D'_\alpha}\|\phi\|^{2}_{L^2(\mathbb{R}^2,\mathbb{H})}
\int_{\mathbb{R}^2}\left|\mathbf{\frac{x}{b}}\right|^{\alpha}|f(\mathbf{x})|^{2}\rm{d}\mathbf{x}\\
\qquad -{D_\alpha}\|\phi\|^{2}_{L^2(\mathbb{R}^2,\mathbb{H})}\int_{\mathbb{R}^2}\left|\mathbf{\frac{x}{b}}\right|^{\alpha}\ln\left|\mathbf{\frac{x}{b}}\right||f(\mathbf{x})|^{2}\rm{d}\mathbf{x}\\
\end{align*}
Now following the  procedure of theorem 4.11 in \cite{gen} we will get desired result.
\end{proof}
\end{theorem}
\begin{remark} Another way to prove Logarithmic uncertainty principle for the two sided ST-QOLCT is by using relation between QFT and ST-QOLCT(\ref{lemma (3.1)}) in  Logarithmic uncertainty
principle for the QFT (see\cite{WLCT} Logarithmic uncertainty principle for the QWLCT ).
\end{remark}
\begin{section}{\bf Conclusions} \ \\
In this paper, first  we establish a relation between two-sided QFT  and two-sided ST-QOLCT. Second, we established some basic properties of the two-sided ST-QOLCT including the Moyal's formula which are proved in \cite{OWN1}. These
results are very important for their applications in digital signal and image processing. Finally,  the uncertainty
principles for the ST-QOLCT such as Donoho-Stark's uncertainty principle, Hardy’s uncertainty
principle, Beurling’s uncertainty
principle, and Logarithmic uncertainty principle are obtained. 
In
our future works, we will discuss the physical significance and engineering background of this paper. Moreover, we will formulate convolution and correlation theorems for the ST-QOLCT.
\end{section}

\end{document}